\documentclass[conference]{IEEEtran}
\IEEEoverridecommandlockouts
\usepackage{caption}
\usepackage{booktabs} 
\usepackage{bbding}
\usepackage{rotating}
\usepackage{balance}
\usepackage{graphicx,graphics}
\usepackage{float}
\usepackage{multirow}
\usepackage{amsmath}
\usepackage{amssymb}
\usepackage[noend,linesnumbered,ruled,vlined]{algorithm2e}
\usepackage[noend]{algpseudocode}
\usepackage{amsthm,amsmath,amssymb}
 
\usepackage{mathrsfs}
\usepackage{subfigure}
\usepackage{latexsym}
\usepackage{lipsum}
\usepackage{color}
\usepackage{booktabs}
\usepackage{comment}
\usepackage{array}
\usepackage{enumitem}
\usepackage{diagbox}
\usepackage{tikz}
\usepackage{pgfplots}
\usepackage{atbegshi}
\usepackage[normalem]{ulem}
\usepackage{fancyhdr}
\usepackage{cite}
\usepackage{tablefootnote}
\usepackage{threeparttable}
\usepackage{colortbl}
\usepackage{hhline}
\def\BibTeX{{\rm B\kern-.05em{\sc i\kern-.025em b}\kern-.08em
    T\kern-.1667em\lower.7ex\hbox{E}\kern-.125emX}}

\newcommand{\ours}{\textsf{Delta} } 
\newcommand{\CBO} {CBO}
\newcommand{\VBO} {VBO}
\newcommand{\one} {VPG}
\newcommand{\two} {CPS}
\newcommand{\detector} {CQD}

\newcommand{\Rmnum}[1]{\uppercase\expandafter{\romannumeral#1}}
\newtheorem{lemma}{Lemma}

\begin{document}

\title{Delta: A Learned Mixed Cost-based Query Optimization Framework}

\author{
    \IEEEauthorblockN{Jiazhen Peng$^{\ast 1 \Delta}$, Zheng Qu$^{\ast 2 \Delta}$,  Xiaoye Miao$^{\ast \S 3 \#}$, 
    Rong Zhu$^{\ddagger 4}$}
    \IEEEauthorblockA{$^\ast$Center for Data Science, Zhejiang University, Hangzhou, China}
    \IEEEauthorblockA{$^\S$The State Key Lab of Brain-Machine Intelligence, Zhejiang University, Hangzhou, China}
    \IEEEauthorblockA{$^\ddagger$Alibaba Group, Hangzhou, China}
    \IEEEauthorblockA{ {\{$^{1}$pengjiazhen, $^{2}$qz,  $^{3}$miaoxy
    \}@zju.edu.cn, $^{4}$red.zr@alibaba-inc.com}}
}


\maketitle

\begingroup
\renewcommand\thefootnote{}\footnote{\noindent
$\Delta$ Equal Contribution.
$\#$ Corresponding authors.
}
\addtocounter{footnote}{-1}
\endgroup



\begin{abstract}
Query optimizer is a crucial module for database management systems.
Existing optimizers exhibit two flawed paradigms:
(1) cost-based optimizers
use dynamic programming with \emph{cost} models but face search space explosion and heuristic pruning constraints;
(2) value-based ones
train \emph{value} networks to enable efficient beam search, but incur higher training costs and lower accuracy.
They also lack mechanisms to detect queries where they may perform poorly. 
To determine more efficient plans, we propose \emph{Delta}, a mixed cost-based query optimization framework that consists of a compatible query detector and a two-stage planner.
Delta first employs a Mahalanobis distance-based detector to preemptively filter out incompatible queries where the planner might perform poorly. 
For compatible queries, Delta activates its two-stage mixed cost-based planner.
Stage $\Rmnum{1}$ serves as a \emph{coarse}-grained filter to generate high-quality candidate plans based on the value network via beam search, relaxing precision requirements and narrowing the search space.
Stage $\Rmnum{2}$ employs a \emph{fine}-grained ranker to determine the best plan from the candidate plans based on a learned \emph{cost} model.
Moreover, to reduce training costs, we reuse and augment the training data from stage $\Rmnum{1}$ to train the model in stage $\Rmnum{2}$.
Experimental results on three workloads demonstrate that Delta identifies higher-quality plans, achieving an average 2.34$\times$ speedup over PostgreSQL and outperforming the state-of-the-art learned methods by 2.21$\times$.
\end{abstract}


\section{Introduction}\label{sec:intro}


The query optimizer is a crucial component of database management systems (DBMSs) and significantly influences query execution performance~\cite{chaudhuri1998overview,selinger1979access,graefe1993volcano,graefe1995cascades,Database}. 
Given an SQL query, the query optimizer aims to determine the optimal execution plan that minimizes the cost of executing
it. 
To measure the execution cost, query execution time, i.e., latency, is often used as a proxy metric~\cite{Database}.
However, due to the vast search space, where there are $O(n!)$ possible plans for a query joining \emph{n} tables, finding the optimal plan is a well-established NP-hard problem~\cite{wang1996complexity}.


Existing solutions address this challenge by incrementally building complete query plans from optimized subplans. The fundamental challenge resides in establishing effective selection criteria that ensure local subplan choices collectively maximize global plan optimality.
Based on the \textbf{subplan selection criteria}, we can categorize existing query optimizers into two classes: cost-based optimizers (CBOs)~\cite{selinger1979access, graefe1995cascades, graefe1993volcano, marcus2021bao, zhu2023lero, yu2022cost, chen2023leon} and value-based optimizers (VBOs)~\cite{marcus2019neo, yang2022balsa, chen2023loger,zhong2024foss}.





\textbf{Cost-based optimizers (CBOs)} can be categorized into traditional methods~\cite{selinger1979access, graefe1993volcano, graefe1995cascades} and learning-based variants~\cite{marcus2021bao,zhu2023lero,chen2023leon,yu2022cost}.
Both share the same paradigm and constraint. They estimate the cost of the subplan, 
and consequently rely on dynamic programming (DP) for plan searching.
However, the explosion of the search space and the limitations of heuristic pruning make it hard to balance plan quality and search efficiency.
Specifically, 
on the one hand, to satisfy the optimality principle of DP~\cite{selinger1979access}, \CBO s need to retain all potentially optimal subplans which have \emph{physical properties}~\cite{graefe1993volcano} like \emph{interesting order}~\cite{selinger1979access}, thus increasing the search space. 
On the other hand, to improve search efficiency,    heuristic space pruning may be relied upon, which leads to sub-optimal or even bad plans.
For example, some \CBO s
reduce the dense tree search space ($O(n!)$) to a left-deep tree search space ($O(2^n)$)~\cite{chaudhuri1998overview}. 
This may cause optimal plans in other parts of the search space to be overlooked.
Moreover, due to storage space limitations, when the number of table joins becomes excessive, \CBO s may resort to random { search} strategies, further compromising the quality of the chosen plan~\cite{postgres}.
In these cases, even if the cost model is accurate, \CBO s are hard to obtain the optimal plans, due to the prohibitive enumeration cost.
As a result, the cost model excels at comparing complete plans, rather than guiding effective exploration through the vast search space.

\textbf{Value-based optimizers (\VBO s)} 
enable efficient exploration of vast search spaces via learned value networks, but their reliance on reinforcement learning (RL) results in higher training costs and lower accuracy compared to the learned cost models.
\VBO s estimate the \emph{overall cost} 
(i.e., \emph{value} 
\footnote{\emph{Overall cost} and \emph{value} (also referred to as \emph{cumulative reward}) come from different perspectives: query optimization and reinforcement learning respectively, yet convey the same meaning: the potential of the subplan developing into the optimal plan.}) 
of each subplan by the \emph{value network}.
The \emph{overall cost} refers to the \emph{minimal} one from the costs of all complete plans that include the given subplan.
It is the \emph{promise} of a subplan to develop into the \emph{optimal} plan, the \emph{real} objective we care about in plan generation.
Thus, simple \emph{best-first search} methods~\cite{Rina1985Best}, such as \emph{beam search}~\cite{Peng1988beam}, can be used to explore plans.
If the value network is \emph{accurate}, 
plans found by \emph{beam search} would be optimal.
 Therefore, the value network is more suitable for searching for promising plans from a vast search space, rather than evaluating complete plans. 
However, the \emph{special} prediction task of the value network makes it have higher training costs and lower accuracy than the learned cost model.

\textbf{Mixed cost-based optimizer}.
By analyzing the pros and cons of the \emph{cost} and \emph{value}, we aim to design a new \emph{mixed} cost-based query optimization framework, named \ours 
to determine more efficient execution plans by leveraging both strengths and complementing each other.

To this end, we design a \emph{two-stage} planner based on the mixed cost.
In stage $\Rmnum{1}$, we design a \emph{value}-based plan generator to generate top-$k$ plans as a candidate plan set from the vast plan space.
The plans are ranked by the values estimated by the \emph{value network}.
Due to the estimation errors of the \emph{value network}, 
the goal of this stage is to act as a \emph{coarse-grained} filter to find a set of candidate plans and ensure that at least one plan in this set achieves high quality.
Compared to \VBO s, which aim to find the best one by precisely ranking all plans,
stage $\Rmnum{1}$ saves the search strength of the \emph{overall cost} and significantly relaxes the precision demand for the value network, 
thereby reducing its training difficulty.
In stage $\Rmnum{2}$, we design a \emph{cost}-based plan selector to perform a \emph{fine-grained} evaluation for candidate plans to determine the best plan based on a learned cost model that predicts the latency of each plan.
Compared to \CBO s, stage $\Rmnum{2}$ leverages the evaluation strength of the \emph{cost} model for a complete plan, and the plan space is narrowed down from $O(n!)$ to $k$, which makes the direct evaluation and comparison feasible.

However, two aspects remain improvable.
\emph{First}, there are two learned models in the mixed cost-based optimizer, causing a high training cost.
To train the value network or the cost model, hundreds or thousands of plans need to be executed to construct the training data.
Each plan can take seconds to minutes to execute~\cite{marcus2021bao,yang2022balsa,chen2023loger}.
Moreover, the plan should be executed 
without other { resource-intensive} activities running on the machine so that the runtime of the plan can accurately reflect the execution cost~\cite{Database}. 
It greatly consumes the system resources.
The high demand on both \emph{quality} and \emph{quantity} of training data is slowing down the wider adoption of learned optimizers. 
To address this challenge, on the one hand, as mentioned above, our two-stage planner reduces the accuracy demand for the value network, thus lowering its training cost. 
On the other hand, we fully utilize the executed plans collected for training the value network in stage $\Rmnum{1}$.
We reuse and augment them for training the cost model in stage $\Rmnum{2}$.

\emph{Second}, all the optimizers, including the mixed cost-based optimizer cannot work well for the arbitrary query.
We formally classify queries where the planner performs poorly as planner-incompatible queries (incompatible queries for short), distinguishing them from compatible queries where the planner maintains stable
effectiveness.
Specifically, our two-stage planner is trained for a specific workload, which makes it prone to significant performance drops when handling queries with different distributions.
To this end, we design a Mahalanobis distance-based detector.
The detector first measures the deviation between the incoming queries and the training workload distribution and then identifies the queries with high deviation as incompatible.
The compatible queries would be optimized by the mixed cost-based planner.
For the incompatible queries, \ours would select a conservative strategy, such as a traditional native query optimizer to generate plans.

In summary, we propose our query optimization framework, \textsf{Delta}. 
It consists of a \emph{two-stage} planner based on the \emph{mixed cost} and a compatible query detector, where the planner consists of a \emph{value}-based plan generator and a \emph{cost}-based plan selector.
The source code referenced in this paper is available on GitHub~\cite{delta}.
Our main contributions
are described below.

\begin{itemize} [itemsep=2pt,topsep=2pt,parsep=0pt,leftmargin=*]

\item We are the \emph{first} to design a \emph{mixed cost}-based planner with a compatible query detector to determine more efficient plans.

\item 
The \emph{mixed cost}-based planner combines a \emph{value}-based generator for high-quality candidates and a learned \emph{cost}-based selector for final plan selection. 
To reduce the training cost, we reuse and augment the training data.

\item The compatible query detector filters out incompatible
queries where the 
mixed cost-based planner 
might perform poorly before planning based on Mahalanobis
distance.

\item  Extensive experiments on three well-known workloads demonstrate the remarkably superior effectiveness of \textsf{Delta}. 
Compared to the traditional optimizer, PostgreSQL, \ours achieves a 2.34$\times$ speedup 
in workload execution.

\end{itemize}

\textbf{Organization}.
Section~\ref{sec:related} reviews related works in query optimization. 
Section~\ref{sec:notation} establishes notation and states the problem we study.
Section~\ref{sec:EASE} presents \textsf{Delta}'s framework, detailing its architecture design and workflow.
Section~\ref{sec:planner} describes the two-stage mixed cost-based planner construction, explaining how it integrates multiple models to search plans and how the models train.
Section~\ref{sec:ADV} reveals the compatible query detection mechanism.
We report detailed evaluation results in Section~\ref{sec:exp} and conclude in Section~\ref{sec:conclusion}.

\section{Related Work}\label{sec:related}
Existing end-to-end query optimizers can be classified into the two following categories based on their own subplan selection criteria.

\subsection{Cost-based Optimizers}
Cost-based optimizers (\CBO s)~\cite{selinger1979access,graefe1993volcano,graefe1995cascades,marcus2021bao,yu2022cost, zhu2023lero, chen2023leon} estimate the cost of individual subplans and employ dynamic programming (DP) to construct plans 
 from the vast plan search space.
Existing \CBO s can be further divided into one-stage ones and two-stage ones.

\textbf{One-stage CBOs}~\cite{selinger1979access,graefe1993volcano,graefe1995cascades, chen2023leon} only take the top-1 plan based on the cost model as the optimal plan,
which makes a high demand for the accuracy of the cost model. 

Traditional optimizers~\cite{selinger1979access,graefe1993volcano,graefe1995cascades} estimate the cost 
using a statistics-based cost model to determine the best plan.
This 
model relies heavily on over-idealized independence and uniformity assumptions, leading to large errors~\cite{zhu2023lero}.
Moreover, the errors propagate exponentially through joins~\cite{Ioannidis1991errors}.

A learning-based optimizer, Leon~\cite{chen2023leon} replaces the traditional cost model with a learning-to-rank model for more accurate ranking of alternative access plans. 
However, it 
overlooks the excessively large training cost, making it impractical~\cite{Lehmann2024behaving}.


\textbf{Two-stage CBOs}~\cite{marcus2021bao, zhu2023lero, yu2022cost} first adopt different strategies that partly guide the plan search process on top of a traditional optimizer like PostgreSQL, to generate diverse plans as candidate plans. They then learn a cost model to estimate the latency or rank of different plans to select the best plan from the candidate plans. 

Bao~\cite{marcus2021bao} considers that the traditional cost model estimation errors affect the choices of \emph{join} or \emph{scan} operators, compromising the generated plan quality.
To correct this,
Bao advises which operations not to use in the whole plan by disabling different subsets of join and scan operations to generate different plans.
Lero~\cite{zhu2023lero} thinks the \emph{cardinality} estimation errors in the cost model ultimately affect the rank of alternative subplans.
It scales the cardinality estimates by different factors at various join levels to affect the rank of alternative subplans to generate diverse plans.
HybridQO~\cite{yu2022cost} specifies different join order leading (i.e. prefix) to guide the plan generation process to generate various plans.
All above plan generators only specify \emph{part} processes (i.e. plan-level operations, cardinality scaling or join order leading, referred to as \emph{hints}) on the traditional \CBO s, which is a \emph{coarse-grained tuning} and thus cannot control the quality of the candidate plan set.


In addition, all \CBO s rely on DP to search plans from the vast plan search space based on the cost model.
The explosion of the search space and the limitations of heuristic pruning make it hard to balance plan quality and search
efficiency.
Even if the cost
model is accurate, \CBO s are hard to obtain the optimal plans, due
to the prohibitive enumeration cost. 

\subsection{Value-based Optimizers}
Value-based optimizers (\VBO s)~\cite{marcus2019neo, yang2022balsa, chen2023loger,zhong2024foss} estimate the overall cost of subplans using the value network 
and explore the entire plan space through beam search to generate the complete execution plan. 
The overall cost indicates the promise of
a subplan to develop into the optimal plan, the real objective we care about in plan generation. 
If the value network is accurate, the plans searched by the beam search would be optimal. 
However, it is more challenging to build the value network than the cost model.
Because obtaining the ground truth of the overall cost of the subplan is harder than the cost.
All the previous works build the value network by reinforcement learning.
They take the top-1 plan estimated by the value network as the execution plan, which makes a high demand for the precision of the value network.

\subsection{Differences from Prior Works}
\ours searches plans by mixed cost models and two-stage search strategy to determine more efficient plans.
Compared to the two-stage \CBO s 
that depend on
traditional optimizers and coarse-grained hints, \ours adopts a value-based plan generator to generate higher-quality candidate plans.
Compared to the \VBO s, \ours modifies the exploration goal from selecting one optimal plan to 
selecting the top-k plans, relaxing the precision requirements for the value network.
Moreover, existing optimizers assume they can handle all situations equally well. This often results in very poor plans that work badly for specific cases.
\ours employs a Mahalanobis distance-based detector to preemptively filter out incompatible queries where the mixed cost-based planner might perform poorly.


\section{Preliminaries}\label{sec:notation}
We first state the notation. 
Then, we show the design goals of the core components.
Table~\ref{tab:notations} summarizes the frequently used notation throughout the paper.

\begin{table}[]     \centering     \footnotesize     \caption{Notation description}     \label{tab:notations}     \begin{tabular}{|p{0.6in}|p{2.4in}|}     \hline     \textbf{Notation} & \textbf{Description} \\ \hline     $Q, P, s$ & the query, the plan and the subplan              \\ \hline     $\mathcal{C}(P), \mathcal{L}(P)$ & the execution cost and the latency of the plan $P$     \\ \hline     $\mathcal{\overline{C}}(Q, s)$ & the overall cost of the subplan $s$ in query $Q$     \\ \hline      $\mathcal{\overline{L}}(Q, s)$ & the overall latency of the subplan $s$ in query $Q$     \\ \hline      $\mathcal{M}, \mathcal{V}$ & the cost model and the value network     \\ \hline          $\mathcal{N}, P_0$ & PostgreSQL’s native optimizer and its plan for $Q$ \\ \hline      $\mathcal{R}(P)$ & the latency of plan $P$ relative to the plan $P_0$     \\ \hline      $\mathcal{G}$ & a candidate plan generator     \\ \hline       $\mathcal{S}_{\mathcal{G}}^k(Q)$ & a set of $k$ candidate plans generated by $\mathcal{G}$     \\ \hline      $P^*, P^{\triangle}$ & the best and worst plan from $\mathcal{S}_{\mathcal{G}}^k(Q)$     \\ \hline      $\mathcal{R}(P^*)$ & the best-case improvement of $S_{\mathcal{G}}^{k}$     \\ \hline      $\mathcal{R}(P^{\triangle})$ & the worst-case risk of $S_{\mathcal{G}}^{k}$     \\ \hline      $\eta(S_{\mathcal{G}}^{k})$ & the ratio of plans outperform $P_0$ in $S_{\mathcal{G}}^{k}$      \\ \hline          \end{tabular}      \vspace{-10pt} \end{table}

\subsection{Notation}

When a user submits a query $Q$, the query optimizer seeks an optimal execution plan to minimize the execution latency, a proxy for the execution cost.
We denote the latency of a query $Q$ or a plan $P$ as $\mathcal{L}(\cdot)$.
The challenge stems from the combinatorial explosion of possible plans.
For $n$-table joins, the search space is more than $\frac{(2(n-1))!}{(n-1)!}$.
These plans yield the same result but differ in the \emph{join order} and \emph{physical operations}, leading to a varied execution cost.

For example, a query joining four tables $(A, B, C, D)$ may have multiple valid subplans for pairwise joins, e.g. \( (A \bowtie B) \bowtie (C \bowtie D) \) vs. \( A \bowtie (B \bowtie (C \bowtie D)) \).
Each subplan contributes differently to the efficiency of the complete plan.
Suboptimal subplan choices can cascade into globally inefficient plans, making subplan selection criteria critical.

To guide subplan selection, there exist two types of models.
One is the cost model $\mathcal{M}(\cdot)$, which estimates the execution cost $\mathcal{C}(s)$ or latency of a subplan $s$:
$$\mathcal{M}: s \rightarrow \mathcal{C}(s)\ or\ \mathcal{L}(s)$$

Another one is the value network $\mathcal{V}(\cdot)$, which predicts the overall cost of the subplan $s$.
The overall cost $\mathcal{\overline{C}}(Q, s)$ (resp. overall latency $\mathcal{\overline{L}}(Q, s)$) refers to the minimal achievable cost (resp. latency) among all complete plans of the query $Q$ that include the given subplan $s$.
To this end, the value network takes in the query $Q$ and the subplan $s$ and outputs the overall cost (or latency) estimate:
$$\mathcal{V}: (Q, s) \rightarrow \mathcal{\overline{C}}(Q, s)\ or\ \mathcal{\overline{L}}(Q, s) $$
Specifically, if $s$ is complete, its overall cost equals its cost.

\begin{table} \centering \footnotesize \caption{Example of overall cost for subplan $s = A \bowtie B$ in query $Q = A \bowtie B \bowtie C \bowtie D$} \vspace{-0.12pt} \label{tab:overall_cost} \begin{tabular}{|c|c|c|} \hline \textbf{Complete Plan $P$} & \textbf{Plan Cost $\mathcal{C}(P)$} & \textbf{Includes $s = A \bowtie B$ ?} \\  \hline $(A \bowtie B) \bowtie (C \bowtie D)$ & 100 & Yes \\  \hline $((A \bowtie B) \bowtie C) \bowtie D$ & 120 & Yes \\  \hline $\ldots$ & $\ldots$ & $\ldots$ \\ \hline $(A \bowtie B) \bowtie (D \bowtie C)$ & 90 & Yes \\  \hline $B \bowtie (A \bowtie (C \bowtie D))$ & 80 & No \\  \hline \end{tabular} \vspace*{-15pt} \end{table} 

To clarify the overall cost, we break it down with an example and contrast it with
subplan cost.
Consider a query \( Q \) joining tables \( A, B, C, D \), and suppose \( s = A \bowtie B \) is a subplan (joining \( A \) and \( B \) first) with cost 20.  
Table~\ref{tab:overall_cost} shows all the complete plans for query $Q$, 
their costs and whether they include subplan $s$.
To compute the overall cost $\mathcal{\overline{C}}(Q, s)$, we notice the plans that include $s$, with costs 100, 120, $\ldots$, and 90.
Suppose the minimal one is 90, the overall cost $\mathcal{\overline{C}}(Q, s)$ is 90.
We know the cost of the subplan $s$ is 20 and 
the best possible plan including $s$ has a cost of 90.
If we select subplans, the overall cost is more useful than the cost of the subplan.
But the overall cost is harder to estimate than the cost.

Traditional one-stage optimizers, like PostgreSQL’s native optimizer  $\mathcal{N}$, iteratively build a plan by dynamic programming based on the cost estimate, producing a baseline plan $P_0$.

One-stage value-based optimizers often train the value network by reinforcement learning and build the plan by beam search based on the overall cost estimate.

Determining the final plan from the vast plan space is difficult.
Some planners divide the plan search process into two stages.
In stage $\Rmnum{1}$, a candidate plan generator \( \mathcal{G} \)  produces $k$ diverse candidate plans $S_{\mathcal{G}}^k(Q) = \{P_1, \ldots, P_k\}$ from the plan search space.
Plan quality is measured by the relative latency, \( \mathcal{R}(P_i) = \frac{\mathcal{L}(P_i)}{\mathcal{L}(P_0)} \).
\( \mathcal{R}(P_i) < 1 \) indicates \( P_i \) outperforms \( P_0 \).  
In stage $\Rmnum{2}$, a selector \( \mathcal{B} \) aims to identify the optimal plan \( P^* = \arg\min_{P \in S} \mathcal{L}(P) \). 
If $\mathcal{B}$ is accurate, the selected plan \( \hat{P} \) matches \( P^* \).  
We can view the one-stage optimizer as a special case of the two-stage optimizer with $k=1$.

\subsection{Problem Statement}\label{sec:problem}
In the two-stage planner, the generated candidate plan set $S_{\mathcal{G}}^{k}(Q)$ \emph{quality} and the ranker $\mathcal{B}$ \emph{accuracy} jointly determine the final plan quality. 
The \emph{quality} of the candidate plan set is crucial, but it lacks its evaluation.
In this paper, we evaluate the quality of the candidate plan set from three key metrics.
\begin{itemize} [itemsep=2pt,topsep=2pt,parsep=0pt,leftmargin=*]

\item The best-case improvement $\mathcal{R}(P^*)$ is defined as the relative latency of the best plan $P^*$ from $S_{\mathcal{G}}^{k}(Q)$, indicating the best achievable performance. The smaller the value, the stronger the optimization potential.

\item The worst-case risk $\mathcal{R}(P^{\triangle})$ is the relative latency of the worst plan $P^{\triangle}$ from $S_{\mathcal{G}}^{k}(Q)$, representing the maximum latency degradation caused by the worst candidate plan. 
If the value is more than 1, the larger 
it is,
the higher the risk.

\item The ratio $\eta({S_{\mathcal{G}}^{k}})$ of the high-quality plans that outperform $P_0$ in $S_{\mathcal{G}}^{k}(Q)$, quantifying the robustness of the candidate set.
The greater the ratio, the higher the probability that the final plan will be better than $P_0$.

\end{itemize}

A candidate plan set can be considered higher-quality if it 
yields a smaller $\mathcal{R}(P^*)$, a smaller $\mathcal{R}(P^{\triangle})$, and a larger $\eta({S_{\mathcal{G}}^{k}})$.

If $\mathcal{R}(P^*)$ is over 1.5, 
the latency of the best plan in the candidate set $\mathcal{L}(P^*)$ will be significantly greater than that of the plan generated by the traditional optimizer like PostgreSQL.
We call such queries the generator $G$-incompatible queries, short for incompatible queries.
Conversely, queries satisfying $\mathcal{R}(P^*) < 1.5$ are termed generator $G$-compatible queries, short for compatible queries.

\textbf{Our goal.}
Given a DBMS, a dataset, and a training workload, our goal is to design a new two-stage query optimization framework to determine a more efficient plan for the user query.
It involves a set of potential problems to solve, including,
(1) designing a plan generator to generate a high-quality candidate plan set;
(2) learning a high-precision model for plan selection with minimal training cost;
(3) designing a compatible query detector for the plan generator.

\section{Framework Overview}\label{sec:EASE}
\begin{figure*}[t]
\center
  \includegraphics[width=0.9\linewidth]{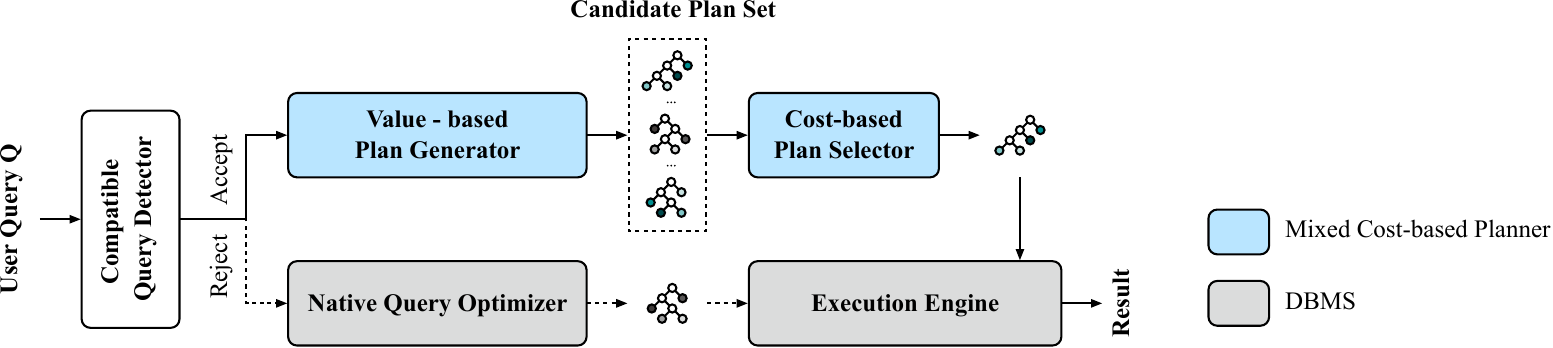}
\caption{The architecture and workflow of \ours}
 \label{fig:overview}
 \vspace*{-0.30in}
\end{figure*}
In this section, to achieve the goals above, we analyze how the proposed framework \ours is designed and introduce its architecture and workflow.

\textbf{Design}. 
As analyzed in Section~\ref{sec:related}, the overall cost (i.e., value) is more suitable than the cost to be the evaluation metric of subplans when generating plans from a vast plan search space.
To this end, to generate a high-quality candidate plan set, we expand the value-based optimizers (\VBO s), such as  LOGER~\cite{chen2023loger} or Balsa~\cite{yang2022balsa} to the value-based plan generator (\one) to select the top-$k$ plans ranked by the value network, rather than the top-1 plan.

Compared to the \VBO s, we find there is a more efficient plan in the top-$k$ plans
than the top-1 plan at a high probability due to the minor error of the value network.
We will prove it in Section~\ref{subsec:one} through formulas that the quality of the top-k candidate plan set is not less than that of the top-1 candidate plan set. Additionally, through experimental observations on LOGER and Balsa, we find that the quality of the top-10 candidate plan set is greater than that of the top-1 candidate plan set in 39.4\% and 91.2\% of the cases, respectively.
Compared to the cost-based plan generators, which apply \emph{coarse-grained hints} on traditional \CBO s and cannot control the hint quality and the heuristic pruning during plan generation, the value-based plan generator can identify the top-$k$ plans according to the estimates of the value network, providing a basis for assessing the quality of the plans.

Due to the simpler prediction task, the cost model is easier to train than the value network.
Therefore, to select the best plan from the candidate plan set, we employ a cost-based plan selector (CPS). The selector learns an execution cost model to predict the latency of each candidate plan and selects the one with minimal estimated latency, rather than using the value network.

To reduce the training cost of the two learned models, we utilize the two-stage planner to \emph{relax the amount demand} for training data of the value network via allocating the accuracy demand.
Meanwhile, we reuse and augment the training data of the value network in stage $\Rmnum{1}$ to train the cost model in stage $\Rmnum{2}$.
Our data augmentation can yield several times diverse yet natural samples by extracting executed plans to supplement training data.
For example, we increase the unique training samples by 8-9$\times$ on the Join Order Benchmark~\cite{leis2015good}
and 2.5-4$\times$ on TPC-DS~\cite{poess2002tpc} and Stack~\cite{marcus2021bao}.

The plan quality generated by \one~and 
 selected by \two~is determined by the accuracy of the learned value network and the cost model.
 Therefore, to guarantee the plan quality and identify compatible queries, we analyze the error sources of the learned models.
We can divide the error causes of learned models into two situations.
1) \textit{In distribution} (i.e., the testing workloads have similar distributions with training workloads): errors come from insufficient training data, noisy training data, loss of features, and so on~\cite{weng2024eraser}.
2) \textit{Out of distribution} (i.e., the testing workloads are out of
distribution with the training workloads): the current learned model without updating lacks the knowledge of the out-of-distribution data, which does not apply to such testing workloads.

For queries in distribution, we employ a two-stage plan search strategy to reduce the accuracy requirements for the value network. Furthermore, in the second stage, we minimize cost model errors through two approaches: data augmentation to supplement training data and heteroscedastic regression loss to mitigate the impact of label noise.
For the queries out of distribution, the predictions of the value network and cost model are quite unreliable. Therefore, such queries should be regarded as \emph{incompatible} queries and we design the \emph{compatible query detector} (\detector) based on the \emph{Mahalanobis distance} between the user query and training data, which is a \emph{model-free} method that does not produce an additional training cost.

\textbf{Architecture and workflow}.
As depicted in Figure~\ref{fig:overview},
\ours consists of a mixed cost-based planner and a compatible query detector
(\detector). 
The mixed cost-based planner consists of a value-based plan generator (\one) and a cost-based plan selector (\two) to determine more \emph{effective} plans with the user queries.
\detector~aims to accept the queries \emph{compatible} for the 
\emph{mixed cost-based planner} and reject the \emph{incompatible} queries.

\ours takes the user query as input and outputs the selected plan for the query as the execution plan to the execution engine of DBMS. 
Specifically, \detector~ rejects the \emph{incompatible} queries whose performance would severely degrade if optimized by the mixed cost-based planner.
Rejected queries are then optimized by
a \emph{conservative} strategy, e.g., the naive optimizer of DBMS. 
The backend can utilize \emph{idle} resources to update (i.e., retrain) the models in the \emph{mixed cost-based planner} to support more queries.
The accepted queries would be optimized by the mixed cost-based planner.
The planner splits the optimal plan search process into two stages.
In the first stage, \one~serves as a \emph{coarse-grained} filter to generate the top-$k$ plans based on a value network as a candidate plan set from a vast plan search space.
In the second stage, \two~determines the optimal plan from the candidates based on a learned cost model, which predicts the latency of each candidate plan.

Next, we will elaborate on how the \emph{mixed cost-based planner} generates efficient plans for \emph{compatible} queries in detail in Section~\ref{sec:planner} and how \detector~detects whether the user queries are \emph{compatible} or not.



\section{Mixed Cost-Based Planner}\label{sec:planner}
\begin{figure*}[th]
\centering
\hspace*{-0.07in}
\subfigure[Plan performance of query 17f in JOB from PG, LOGER, and Balsa]{
\label{fig:example}
\raisebox{-0.2cm}{\includegraphics[width=0.49\linewidth]{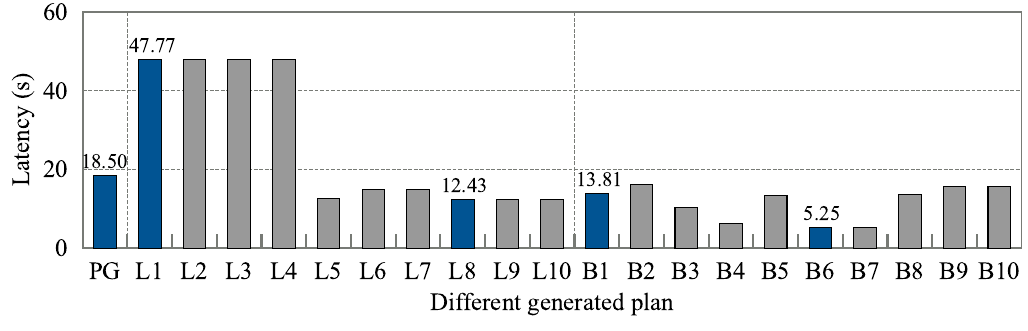}}}
\subfigure[Frequency of best plan positions within the top-10 plans generated by VBOs.
]{
\label{fig:statistics}
\raisebox{-0.2cm}{
\includegraphics[width=0.245\linewidth]{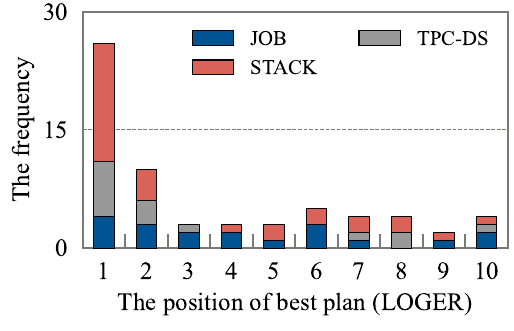}\hspace*{-0.027in}
\includegraphics[width=0.245\linewidth]{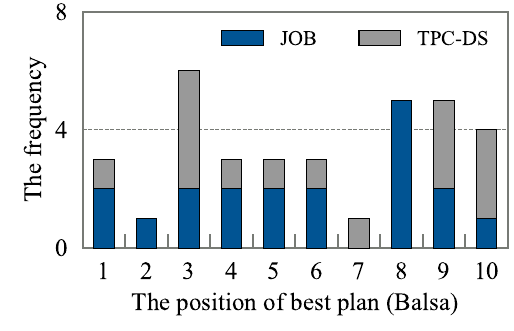}
}
}
\vspace{-0.28pt}
\caption{The example and statistics of top-1 VS top-$k$ plans generated by \VBO s}
\end{figure*}

For an incoming compatible query, the mixed cost-based planner divides the query optimization process into two stages, which fully takes advantage of the value network and cost model to determine a more effective plan.  

\subsection{Value-based Plan Generator} \label{subsec:one}
In the first stage, 
we aim to design a plan generator to generate a high-quality candidate plan set from the vast plan search space for the input query.

As described in Section~\ref{sec:related}, we find the value network is more suitable than the cost model for guiding the plan generator to focus on promising regions in the vast plan search space, but not suitable for determining the final optimal plan.
Furthermore, we find there is a more efficient plan in the top-$k$ plans generated by \VBO s
than the top-1 plan at a high probability due to the estimation error of the value network.

For example, 
Figure~\ref{fig:example} shows the performance (i.e., the 
query plan execution time) of the \emph{state-of-the-art} (SOTA) \VBO s including LOGER~\cite{chen2023loger} and Balsa~\cite{yang2022balsa}, compared to the traditional optimizer PostgreSQL (PG for short).
The query 17f is from the JOB workload~\cite{leis2015good}.
L$i$ (resp. B$i$) represents the $i$-th plan generated by LOGER (resp. Balsa), $i=1,2,\ldots, 10$. Particularly, L1 (and B1) is the plan selected by LOGER (and Balsa).
One can observe that there exists a better plan (i.e., L8 and B6 respectively) of query 17f in the top-10 plans than the plan selected by \VBO s (i.e., L1 and B1 respectively).


To this end, we expand the existing \VBO s  to the 
\one\   
to generate top-$k$ plans estimated by the value network as the candidate plans set, rather than the top-1 plan.
The best-case improvement of the top-$k$ plans is no lower than that of the top-1 plan, as claimed in the following lemma with proof.

\vspace{-0.2pt}
\begin{lemma}
Let \( S_{\mathcal{V}}^k = \{P_1, P_2, \ldots, P_k\} \) be the set of top-\(k\) plans generated by the value-based plan generator, and \( S_{\mathcal{V}}^1 = \{P_1\} \) be the top-1 plan. For any \( k \geq 1 \), the best-case improvement satisfies:  
\[
\mathcal{R}(P^*) \leq \mathcal{R}(P_1),
\]  
where \( P^* = \arg\min_{P \in S_{\mathcal{V}}^k} \mathcal{R}(P) \), and \( P_1 \) is the best plan in \( S_{\mathcal{V}}^1 \).
\end{lemma}

\begin{proof} 
\renewcommand{\qedsymbol}{}
By definition, the best-case improvement of \( S_{\mathcal{V}}^k \), denoted as \( \mathcal{R}(P^*) \), can be represented as:  
\[
\mathcal{R}(P^*) = \min_{P \in S_{\mathcal{V}}^k} \mathcal{R}(P) = \min \big\{ \mathcal{R}(P_1), \mathcal{R}(P_2), \ldots, \mathcal{R}(P_k) \big\}.
\]  
Since \( P_1 \) is contained in \( S_{\mathcal{V}}^k \), it directly follows that:  
\[
\min \big\{ \mathcal{R}(P_1), \mathcal{R}(P_2), \ldots, \mathcal{R}(P_k) \big\} \leq \mathcal{R}(P_1).
\]  
Thus, \( \mathcal{R}(P^*) \leq \mathcal{R}(P_1) \).
\end{proof}



Furthermore, we report the statistical performance of
plan quality
(across three query workloads with 64 queries) 
generated by LOGER and Balsa in Figure~\ref{fig:statistics}.
It counts the plan position (from top-1 to top-10) on which it is the most efficient.
It shows the position (from top-1 to top-10) where the most efficient plan appears.
The ratio of the top-1 plans generated by LOGER being the best ones in the top-10 plans is 40.6\%, and that of Balsa is only 8.8\%. 


To this end, we utilize the value-based plan generator 
to generate the top-$k$ plans as the candidate plans set. 
Compared to the cost-based plan generators, which apply \emph{coarse-grained hints} on traditional \CBO s and cannot control the hint quality and the heuristic pruning during plan generation, the value-based plan generator can identify the top-$k$ plans according to the estimates of the value network, providing a basis for assessing the quality of the plans.
To this end, the quality of plans generated by \one\ is higher than \CBO s, and we will demonstrate it in Section~\ref{subsec:horses}.

\textbf{Top-$k$ plans generation}. 
 Given the query $Q$, the plan generator can generate the complete plan in a bottom-up way by following steps.
First, the generator needs to start from an initial subplan $s_0$, which means an empty plan with no operator is specified. 
Then, it chooses a new operation implemented on the current subplan $s_i$.
The operation can be a scan operator implemented on a table or a join operator that joins two joinable subplans.
Once the operator is determined, the current subplan $s_i$ is transformed into the next subplan $s_j$ and then the generator will choose the next operator.
Finally, the plan is completed when the all tables in $Q$ are joined.

Given the trained value network $V$, 
\one\ generates the top-$k$ plans in a bottom-up way through beam search.
Beam search algorithm explores the promising regions of the plan search space in a limited set, according to the rank of promising values of subplans, which is estimated by the value network.   

\begin{algorithm}[t]\small
\vspace{-0.3pt}
\caption{Top-$k$ Plans Generation by Beam Search}
\label{alg:top-k}
\DontPrintSemicolon
\LinesNumbered
\SetNlSty{Large}{}{:}
\KwIn{a user query $Q$, a trained value network $V$, a beam width $b$, and the number of generated plans $k$}
\KwOut{the generated top-$k$ plans set $S=\{P_1, P_2, \ldots, P_k\}$}
$S \gets \varnothing, s_0\gets \varnothing, \mathcal{B}\gets \{\langle s_0, 0\rangle\}$ \tcc*{initialization}
\While{\emph{$|S|<k$ and $|\mathcal{B}|\neq0$}}
{
    $\langle s_i, \overline{l_i}\rangle \gets$ remove the top element from $\mathcal{B}$\;
    \If{\emph{$s_i$ is a complete plan}}{
        add $s_i$ to $S$\;
    }
    \Else{
        $S_j \gets$ explore all the next states of $s_i$\;
        \For{\emph{each} $s_j \in S_j$}{
            $\overline{l_j} \gets V(Q, s_j)$  \tcc*{compute overall latency}
            insert $\langle s_j, \overline{l_j}\rangle$ to $\mathcal{B}$\ in  ascending order of overall latency\;
        }
        \If{$|\mathcal{B}|>b$}{
            leave the top $b$ elements and remove others from $\mathcal{B}$\;
        }
    }
}
\Return $S$
\end{algorithm}

Algorithm~\ref{alg:top-k} describes how beam search
uses the trained value network $V$ to search the top-$k$ optimal plans for the user query $Q$ in the vast plan search space.
The algorithm maintains a set with the max size $b$ (i.e., beam width in beam search).  It stores the promising subplans sorted by the subplan's overall latency   (estimated by the value network $V$) in ascending order.
It gets subplans sequentially from the set and explores them in a bottom-up way, until generating $k$ complete plans.

We first empty a \emph{set} $S$ to store the top-$k$ complete plans to be returned, an \emph{initial subplan} $s_0$ with its overall latency 0, and a \emph{set} $\mathcal{B}$ to store the \emph{promising subplans} with the overall latency values at the beginning (line 1). 
We then explore the subplans sequentially, until generating $k$ complete plans (lines 2-12).
We get and remove the best subplan $s_i$ with \emph{minimal} overall latency $\overline{l_i}$ from the top of $\mathcal{B}$ (line 3).
If the subplan $s_i$ is complete, we add it to the plans set $S$ (lines 4-5). Otherwise, we continue exploring from the subplan $s_i$  and updating the set $\mathcal{B}$ (lines 6-12). 
We explore all the next new subplans of $s_i$ by adding an operation on $s_i$ and the new subplans form a set $S_j$.  
Then, we compute the overall latency of each new subplan based on the value network and add them into the set $\mathcal{B}$ in the \emph{ascending} order of overall latency (lines 7-10).
$\mathcal{B}$ is sorted by the overall latency, which determines the exploring order.
If the size of $\mathcal{B}$ is more than $b$, the excessive subplans with high overall latency values will be removed (lines 11-12).
The process repeats until $k$ complete query plans are obtained.
Finally, it returns the generated top-$k$ plans.

\textbf{The training of the value network}. 
\ours builds on the value-based optimizer (\VBO), LOGER~\cite{chen2023loger} or Balsa~\cite{yang2022balsa}.
Next, we present an overview of the universal training process of the value network in the \VBO.

Obtaining the ground truth of the overall latency of the given subplan in the value network is more challenging than the latency of that in the cost model.
The reason is obtaining the overall latency requires traversing the plan search space to get the actual execution costs of plans that include the given subplan, which is impractical.
To alleviate the problem, the overall latency can be approximated as the best latency seen so far of a plan that includes the given subplan by reinforcement learning (RL), instead of the global optimal one.

The training process is iteratively observing the execution of diverse plans for each query to obtain feedback and then utilizing it to update the value network. 
As the loop iterates, 
the generated plans are more efficient, 
the approximation of the overall latency gets better, 
and the value network improves.

Specifically, for each iteration, for each query of the training queries $Q_{train}$, we explore different plans and execute them to obtain their latency values.
The executed plans are stored in a plan pool.
The training of the value network would generate and execute different plans for each query. 
Each executed plan $P$ is recorded in the plan pool in the format $(Q, P, \mathcal{L})$, where $Q$ and $\mathcal{L}$ correspond to the query and latency, respectively.
The executed plan dataset will be reused in the training of the learned cost model in the second stage.
For each $(Q, P, \mathcal{L})$, we generate subplans of the plan $P$.
For each subplan $s$, we can get its overall latency $\mathcal{L}$ and store $((Q, s), \mathcal{L})$ in a hash lookup table. 
If $(Q, s)$ exists in the hash table, and the $\mathcal{L}$ is lower than the $\mathcal{L'}$ 
pre-existing in the table, we can update the overall latency with the better value $\mathcal{L}$.
 The value network can be updated with the new data.

\subsection{Cost-based Plan Selector}
\begin{figure}[t]
\center
\includegraphics[width=0.85\linewidth]{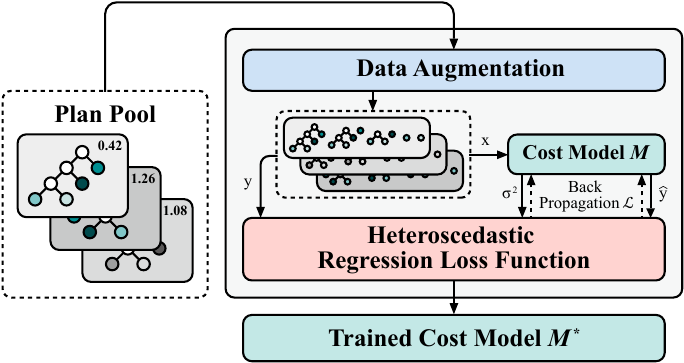}
  \caption{The training of the cost model $\mathcal{M}$}
 \label{fig:model}
 \vspace*{-0.25in}
\end{figure}

In the first stage, \one\ narrows down the range of optimal plan selection from the vast plan search space to $k$ candidate plans based on the value network and beam search.
To this end, in the second stage, we design a cost-based plan selector \two\ to train a high-precision cost model $M$ with minimal extra training cost to guide the planner in determining the optimal plan from the candidates.

\two\ adopts the classical cost model $\mathcal{M}$, which predicts the latency
of a plan to determine the plan with minimal estimated value as the optimal plan.
Figure~\ref{fig:model} shows the training of the cost model $\mathcal{M}$. 
Since the system resource consumption of collecting training data is non-trivial, 
we \emph{reuse} the executed plans in the \emph{plan pool} generated by the training of the value network in stage $\Rmnum{1}$ to train a high-precision cost model $\mathcal{M}$.
Data reusing ensures that the cost model $\mathcal{M}$ does not require additional training data and thus greatly saves the training cost.
Based on the naive prediction task, we can utilize \emph{data augmentation} to yield several times diverse yet natural samples by extracting executed plans to supplement training data.
Then, to reduce the influence caused by label noise, 
we adopt a heteroscedastic regression loss function~\cite{Weihang2023Regression} to train the model, which models the noise of each sample and thus improves the robustness of estimates.

It is worth noting that 
there are two alternative prediction tasks for optimal plan selection.
The first uses a value network to predict the overall cost of
a given query and plan~\cite{marcus2019neo,yang2022balsa,chen2023loger}.
As discussed before, the value network is not suitable for determining the optimal plan 
due to lower accuracy.
The second adopts pair-wise learning-to-rank models~\cite{zhu2023lero, chen2023leon}, which predicts the ranking of two plans sorted by their latency.
The learning-to-rank model predicts the real object the stage cares about and the pair-wise ranking is essentially the classification task, which is easier to fit than the regression task.
However, their training 
remains inefficient,
as it cannot 
fully leverage
the executed plans to enhance accuracy.

\textbf{Data augmentation}. 
The high training cost of a learned model in query optimization 
comes from
\emph{data collection} and \emph{model training}.
Data collection requires hundreds or thousands of plans to execute, 
which 
greatly consumes system resources. 

Many previous studies~\cite{marcus2021bao,yu2022cost,zhu2023lero} suggest that the cost model can learn from the executed plans of historical user queries.
Therefore, there's no need to consider the overhead of training data collection. 
However, it is impractical.
The core issue lies in the inherent instability of system environments during query execution.
Fluctuations in system states cause the execution latency of a plan to reflect its true execution cost poorly. 
This inconsistency introduces large noise into the training data, ultimately undermining the model's accuracy and reliability.
For instance, a simple query might exhibit abnormally long runtime under a heavy system load, leading the model to misinterpret its intrinsic efficiency. Such discrepancies between observed latency and actual execution cost make historical execution data unreliable for training robust cost models.

The runtime of distinct plans could be used as a measure of their execution cost only when these plans are executed under a hypothesized stable state of the system environment, 
such as the consistent initial disk cache state and no other resource-intensive activities running in the system ~\cite{Database}.
It makes training data collection expensive. 

To this end, we make full use of the executed plans to train the model by augmenting training data.
DBMS not only records the latency of each executed plan but also records the execution time of each subplan.
For a sample $(Q, P, \mathcal{L})$ in the executed plan pool collected in the first stage, we can extract subplans and their direct latency values from the executed plan $P$.
Each subplan $s$ with its execution time $\mathcal{L}(s)$ can constitute a unique training sample $(s, \mathcal{L}(s))$. 
This generates many natural and diverse samples, greatly enriching the quantity and diversity of the training data, and improving the accuracy of the model at a negligible cost.
As a result, we utilize such a data augmentation strategy to increase the unique training samples by 8-9$\times$ on the Join Order Benchmark (JOB)~\cite{leis2015good}
and 2.5-4$\times$ on TPC-DS~\cite{poess2002tpc} and Stack~\cite{marcus2021bao}.

\textbf{Loss function}. 
The latency of the plan, as the label of the cost model, inevitably has aleatoric noise due to the minor disturbance,
even though we maintain the consistent initial disk cache state, and no other processes interference.
To reduce the noise, one naive method is to run the same plan multiple times and take the average value as the label.
However, this method consumes too many system resources. Thus, we adopt the heteroscedastic regression loss function~\cite{Weihang2023Regression}  to reduce the influence caused by the aleatoric noise.

Different from classical regression, which assumes all of the modeling errors have the same variance, i.e., the variance is considered to be a constant and will be ignored, heteroscedastic regression assumes that the variance values of errors are different and related to the data points.  
The aleatoric noise of training data can be considered to be heteroscedasticity~\cite{Kendall2017What}.

Therefore, our cost model $\mathcal{M}$ predicts the latency $\mathcal{L}(P)$ of the given plan $P$ and the noise $\epsilon(P)$. It can be formulated as $$\mathcal{M}: P \rightarrow (\mathcal{L}(P),\epsilon(P))$$  
We denote the number of training samples as $N$, the latency of $i$-th sample as $y_i$, 
and the corresponding estimated latency, the estimated noise
of the model as $\hat{y_i}$, $\sigma_i^2$, respectively.
The heteroscedastic regression loss function $\mathscr{L}$ is calculated as
$$\mathscr{L}=\frac{1}{N} \sum_{i=1}^{N}\left\{\frac{\left(y_{i}-\hat{y}_{i}\right)^{2}}{2 \sigma_{i}^{2}}+\frac{\ln \sigma_{i}^{2}}{2}\right\}$$ 

The object of the cost model $\mathcal{M}$ training is to minimize the heteroscedastic regression loss function $\mathscr{L}$.
It makes the prediction more robust. Intuitively, the first term of $\mathscr{L}$, $\sigma_{i}^{2}$ represents the data sample noise. That is, the sample with larger noise indicates lower quality, thus we reduce the weight coefficient of such samples by $1/\sigma_{i}^{2}$. 
This makes the model $\mathcal{M}$ pay more attention to fitting samples with smaller noise and reduce the impact of samples with larger noise. The second term is a penalty term that prevents the model from constantly increasing $\sigma_{i}^{2}$ to reduce the loss.

\textbf{Model details}.
We combine the ways adopted by Bao~\cite{marcus2021bao} and Lero~\cite{zhu2023lero} to featurize plan trees to vector trees by encoding each node in the plan trees.
Each node in a plan tree is featurized into a vector containing: (i) a one-hot encoding of the operator, (ii) the min-max normalized cardinality and cost estimates, and (iii) the one-hot encoding sum of the relations involved in the subplan rooted at the node.  
The featurized query plan tree is passed through three blocks consisting of tree convolution~\cite{Mou2016treeconv,marcus2019neo}, the layer normalization~\cite{Hinton2016Layer}, and ReLU activation layer~\cite{Xavier2011relu}. 
Dynamic pooling~\cite{Mou2016treeconv} replaces the activation layer of the last block to flatten the tree structure into a single vector.
The single vector is passed through two blocks of two fully connected layers with the same structure but different weights to predict the latency and noise, respectively.

\section{Compatible Query Detector} \label{sec:ADV}

To ensure robust optimization, \ours employs a Compatible Query Detector (CQD) to identify queries where the mixed cost-based planner might underperform.

Recall the definition of compatible queries in Section~\ref{sec:problem} and the error sources analysis of the mixed cost-based planner in Section~\ref{sec:EASE}.
A query is deemed compatible with the mixed cost-based planner by CQD if its feature distribution aligns with the training data distribution.
This ensures the planner’s value network and cost model can reliably predict execution plan quality. 
Conversely, out-of-distribution (OOD) queries, those deviating from the training data, cause degraded model predictions, leading to suboptimal plans.
CQD rejects such incompatible cases to safeguard optimization robustness. 
It intercepts OOD queries by computing the Mahalanobis distance~\cite{Mahalanobis1936OnTG,de2000mahalanobis} between query features and the training distribution. 
Unlike Euclidean distance, the Mahalanobis distance is unitless, scale-invariant, and accounts for variable correlations, making it widely adopted for OOD
~\cite{lee2018simple} and outlier detection~\cite{ramaswamy2000efficient}.

The detection process begins with generating a representative baseline plan. For an input query \( Q \), the native optimizer produces an execution plan \( P_0 \), balancing efficiency and representativeness. This design aligns with the value network’s initial training on plans generated by the native optimizer. 
Next, the query-plan pair \( (Q, P_0) \) is encoded into a high-dimensional feature vector \( \vec{x} \) using the same method as the value network. 
Each training sample $(Q', P')$ from the training workload is similarly encoded and all of them are combined into
a matrix $\mathbf{X}$.
The Mahalanobis distance between a test sample \( \vec{x} \) and the training distribution \( \mathbf{X} \) is calculated as:  
\[
    \mathcal{D}_M(\vec{x},\;\mathbf{X}) = \sqrt{(\vec{x}-\vec{\mu})^{\mathrm{T}}\mathbf{\Sigma}^{-1}(\vec{x}-\vec{\mu})},
\]
where \( \vec{\mu} \) is the mean vector and \( \mathbf{\Sigma} \) the covariance matrix of \( \mathbf{X} \). A larger distance indicates lower similarity. A threshold-based indicator function \( \mathbb{I} \) determines acceptance or rejection:  
\[
    \mathbb{I} \rightarrow
    \begin{cases}
        1 \, (\text{Accept}) & \text{if } \mathcal{D}_M(\vec{x},\;\mathbf{X}) \leqslant \gamma, \\
        0 \, (\text{Reject}) & \text{if } \mathcal{D}_M(\vec{x},\;\mathbf{X}) > \gamma.
    \end{cases}
\]

Queries exceeding \( \gamma \) are routed to the native optimizer. 
The threshold \( \gamma \) controls the trade-off between reliability and optimization coverage.
A smaller $\gamma$ enforces a conservative strategy, rejecting more queries and routing them to the native optimizer.
This ensures robustness
but limits potential performance gains from the mixed cost-based optimizer.
A larger $\gamma$ relaxes acceptance criteria, allowing the learned optimizer to handle a broader range of queries. This increases the risk of including out-of-distribution queries, which may compromise reliability.
This trade-off enables the system to balance optimization and stability according to practical requirements.

By 
filtering incompatible queries, CQD ensures the planner operates only on in-distribution data where 
learned models are reliable, balancing plan quality and robustness.

\section{Experimental Evaluation} \label{sec:exp}


In this section, we conduct experiments to answer the following questions.
The experimental setup is described in Section~\ref{sec: setup}.
\begin{itemize}[itemsep=2pt,topsep=2pt,parsep=0pt,leftmargin=*]
    \item How much can \textsf{Delta} enhance final plan quality and training efficiency of query optimizers? (Section~\ref{subsec:overall})
    
    \item Does the value-based planner generate higher-quality candidate plan set than cost-based ones? (Section~\ref{subsec:horses})
    
    \item Does the cost model outperform the value network in selecting the optimal plan from candidate plans? (Section~\ref{subsec:horses})
    
    \item 
    Can CQD correctly identify the most compatible queries without negative effects?
    (Section~\ref{subsec: 7.4})
    
    \item What are the benefits of \textsf{Delta}'s different design choices? (Section~\ref{subsec:ablation_and_parameter})

    \item How do the parameters (e.g., $k$, $\gamma$) impact the overall performance. (Section~\ref{subsec:ablation_and_parameter})
\end{itemize}

\subsection{Experiment Setup}~\label{sec: setup}
\textbf{Workloads}. 
We conduct experiments on three workloads from three widely used benchmarks.
\vspace{-3pt}
\begin{itemize}[leftmargin=0pt,itemindent=15pt]
    \item {\bf Join Order Benchmark (JOB):}
    This is a query workload designed by Leis~\cite{leis2015good}, consisting of 113 queries generated from 33 query templates. 
    The dataset is based on the real-world Internet Movie Database (IMDB), sized at 3.6GB and containing 21 tables.
    For the query split, we follow Balsa~\cite{yang2022balsa}, randomly selecting 19 queries for testing and using the remaining 94 for training.
    This dataset is used to evaluate optimizer performance on real-world data.


    \item {\bf TPC-DS:}
    This benchmark includes both a database and a query generator~\cite{poess2002tpc}.
    We set the scale factor to 10 to generate approximately 10GB of data, including 26 tables.
    For the training-test split of the queries, we follow LOGER~\cite{chen2023loger}, selecting 5 templates as the test set and the other 15 templates as the training set.
    Three queries are generated for each template to form the test and training workloads.
    This workload evaluates optimizer performance on generated datasets.

    \item{\bf Stack:}
    This large and real dataset is designed by Bao~\cite{marcus2021bao}.
    It contains over 100GB of data from StackExchange websites (e.g., StackOverflow.com) and includes 11 tables.
    Some tables have composite primary keys.
    Similar to LOGER~\cite{chen2023loger}, we use 6 query templates with 5 queries each for testing, and 8 templates with 10 queries each for training.
    This workload is used to evaluate query optimizers on large-scale datasets.
\end{itemize}
\textbf{Baselines}. 
We choose the open-source DBMS \textbf{PostgreSQL}~\cite{postgres} as the baseline for the traditional \CBO.
Additionally, we use three state-of-the-art (SOTA) two-stage learned \CBO s: \textbf{Bao}~\cite{marcus2021bao}, \textbf{HybridQO}~\cite{yu2022cost}, and \textbf{Lero}~\cite{zhu2023lero}.
For \VBO s, we choose \textbf{LOGER}~\cite{chen2023loger} and \textbf{Balsa}~\cite{yang2022balsa}.
All optimizers are implemented following the provided source codes. 
For each \VBO, we deploy the \textsf{Delta} architecture on top, referred to as \textbf{LOGER$\textsf{-}\Delta$} and \textbf{Balsa$\textsf{-}\Delta$}, to assess the adaptability of the \textsf{Delta} architecture to different underlying \VBO \ structures.
\textbf{Metrics}. 
We utilize three primary metrics to evaluate the performance from various perspectives.
\begin{itemize}[leftmargin=0pt,itemindent=15pt]
    \item \textbf{Workload relative latency (WRL)} evaluates the performance of a learned query optimizer against PostgreSQL at the workload level. 
Let \( W = \{Q_1, Q_2, \ldots, Q_n\} \) denote a workload of \( n \) queries. For each query \( Q_i \), let \( P_i \) be the plan selected by the learned optimizer and \( P_{0_i} \) be the plan generated by PostgreSQL’s native optimizer \( \mathcal{N} \). 
WRL computes the workload total latency ratio of the learned optimizer to PostgreSQL:  
\[
WRL = \frac{\sum_{i=1}^n \mathcal{L}(P_i)}{\sum_{i=1}^n \mathcal{L}(P_{0_i})}
\]  
Moreover, \( 1/\text{WRL} \) represents the \emph{Speedup}, 
quantifying how much faster the learned optimizer executes the workload.

    \item   \textbf{Geometric mean relevant latency (GMRL)}~\cite{Yu2020join, chen2023loger,chen2023leon,chen2023base} reflects the average query-level execution performance relative to PostgreSQL.
    It is defined as the geometric mean of the relative latency values across the workload:
\[
GMRL = \left( \prod_{i=1}^n \mathcal{R}(P_i) \right)^{\frac{1}{n}} = \left( \prod_{i=1}^n \frac{\mathcal{L}(P_i)}{\mathcal{L}(P_{0_i})} \right)^{\frac{1}{n}}
\]

     \item \textbf{Training time $T_{t}$}  measures the time required for each learned optimizer to reach a stable state, reflecting its training cost.
     It includes both the model training time and the time for collecting training data. 
     The latter usually dominates, as it involves executing thousands of plans sequentially in a relatively stable system environment (e.g., quasi-consistent initial cache state, no other running resource-intensive processes).
\end{itemize}

\textbf{Parameters and environments}.
For baselines, we use the default setting in their original papers.
For \textsf{Delta}, we set the number of candidate plans $k$ to 10,
beam width $b$ to 20, and threshold $\gamma$ to 500.
Experiments are conducted on an Intel Core 2.80GHz server with 376GB RAM and three NVIDIA RTX 2080 Ti GPUs for model training and inference, running Ubuntu 18.04.
PostgreSQL 12.5 is used, configured with 96GB shared buffers, 192GB cache, and GEQO disabled.
We apply the \emph{pg\_hint\_plan}~\cite{pghintplan} plugin to specify the plan generated by \textsf{Delta}. 
To ensure a relatively consistent initial disk cache state for a fair comparison, 
we preload data using \emph{pgfincore}~\cite{pgfincore} and \emph{pg\_prewarm}~\cite{pgprewarm} for the Linux and PostgreSQL caches respectively before training and testing.
Each experiment is repeated three times, and the average result is reported to mitigate randomness.

\subsection{Overall Performance} \label{subsec:overall}
To assess \textsf{Delta}'s overall performance, 
we test each optimizer extensively on the three workloads.
All learned optimizers are trained on the training set until convergence or a 48-hour limit, then evaluated on an unseen test set to measure their actual performance.

\textbf{Workload latency analysis.}
Table~\ref{tab:final} reports \emph{WRL}, \emph{GMRL}, 
and $T_{t}$ of \textsf{Delta} and other optimizers on three test workloads.
Optimal results are in \textbf{bold}.
\textcolor{gray}{TLE} means that the optimizer cannot finish training within 48 hours.
Balsa fails to handle composite key joins in \emph{Stack}, leading to poor plans and excessive training time.

As shown in this table, with adequate training, \textsf{Delta}  outperforms all baselines in latency across all workloads.
Compared to PostgreSQL, LOGER$\textsf{-}\Delta$ achieves
speedups of 3.57$\times$, 2.22$\times$, and 2.22$\times$ on JOB, TPC-DS, and Stack, 
while Balsa$\textsf{-}\Delta$ yields 2.50$\times$ and 1.33$\times$ on JOB and TPC-DS.
Against two-stage learned \CBO s (Bao, Lero, HybridQO), \textsf{Delta} also improves plan latency.
Specifically, on the JOB, TPC-DS, and Stack, LOGER$\textsf{-}\Delta$ delivers average speedups of 3.36$\times$, 1.97$\times$, and 1.71$\times$,
while Balsa$\textsf{-}\Delta$ reaches 2.35$\times$ and 1.18$\times$ on JOB and TPC-DS.
When compared to \VBO s, \ours also demonstrates significant performance gains.
LOGER$\textsf{-}\Delta$, relative to LOGER, achieves speedups of 1.14$\times$, 1.31$\times$, and 1.51$\times$ on JOB, TPC-DS, and Stack, 
and Balsa$\textsf{-}\Delta$ outperforms Balsa by 1.55$\times$ and 4.21$\times$ on JOB and TPC-DS.

These improvements 
stem from \textsf{Delta}'s innovative mixed cost-based planner.
It combines a value network for exploring high-quality candidate plans and a cost model for accurately selecting the optimal one.
This synergy enhances the quality of the final plan and leads to a significant improvement in query execution speed.
Furthermore, CQD helps filter out 
 incompatible queries, 
preventing significant performance regressions.

\begin{table}
    \centering
    \scriptsize
    \footnotesize
    \setlength{\tabcolsep}{2pt}
    \caption{Performance comparison on three test workloads}
    \label{tab:final}
    \setlength{\extrarowheight}{0pt}
    \setlength{\aboverulesep}{0pt}
    \setlength{\belowrulesep}{0pt}
    \arrayrulecolor{black}
    \begin{tabular}{|c|c|c|c|c|c|c|c|c|c|} 
\midrule
\multirow{2}{*}{Optimizer} & \multicolumn{3}{c|}{$\emph {JOB}$}             & \multicolumn{3}{c|}{$\emph{TPC-DS}$}          & \multicolumn{3}{c|}{$\emph{Stack}$}                                                                \\ 
\cline{2-10}
                           & WRL           & GMRL        & $T_{t}$       & WRL           & GMRL         & $T_{t}$      & WRL                                     & GMRL                                    & $T_{t}$       \\ 
        \hline
        PostgreSQL                                                                                                                                    & 1.00 & 1.00 & —          & 1.00 & 1.00  & —              & 1.00 & 1.00  & —                                            \\ 
        \midrule
        Bao                                                                                                                                           & 0.75 & 1.23 &\textbf{0.43}                & 1.05 & 1.04 &0.22          & 0.86 & 0.97 & \textbf{0.39}                                        \\ 
        Lero                                                                                                                                          & 1.08 & 1.17  & 1.34       & 0.62 & 0.64 & 1.07          & 0.48 & 0.93 & 15.34                                       \\ 
        HybridQO                                                                                                                                      & 0.99 & 0.99 & 0.48       & 0.99 & 0.99 &\textbf{0.11}          & 0.97 & 0.98 &  0.46                                        \\ 
        \midrule
        LOGER                                                                                                                                         & 0.32 & 0.60 & 4.72       & 0.59 & 0.48 & 3.09         & 0.68 & 1.02 & 8.83                                       \\ 
        Balsa                                                                                                                                         & 0.62 & 1.34 & 11.37      & 3.16 & 1.25 &  10.00         & \textcolor{gray}{TLE}  & \textcolor{gray}{TLE}  & —                                           \\ 
        \hline
        LOGER$\textsf{-}\Delta$                                                                                                                                   & \textbf{0.28} & \textbf{0.56} & 1.88      & \textbf{0.45} & \textbf{0.40}  & 1.06         & \textbf{0.45} & \textbf{0.74}   & 6.10                                       \\ 
        Balsa$\textsf{-}\Delta$                                                                                                                                 & 0.40 & 0.74  & 10.61        & 0.75 & 0.63  &9.33   &\textcolor{gray}{TLE} & \textcolor{gray}{TLE}   & —                                            \\
        \midrule
    \end{tabular}
\vspace{-10pt}
\end{table}
\begin{figure*}[th]
\centering
\subfigure[LOGER and LOGER$\textsf{-}\Delta$ on $\emph {JOB}$]{

    \begin{minipage}[t]{0.30\linewidth}
        \includegraphics[width=1\linewidth]{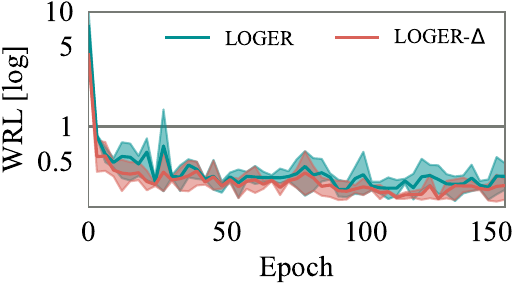}
    \end{minipage}
    
}
\subfigure[LOGER and LOGER$\textsf{-}\Delta$ on $\emph {TPC-DS}$]{
    \begin{minipage}[t]{0.30\linewidth}
        \includegraphics[width=1\linewidth]{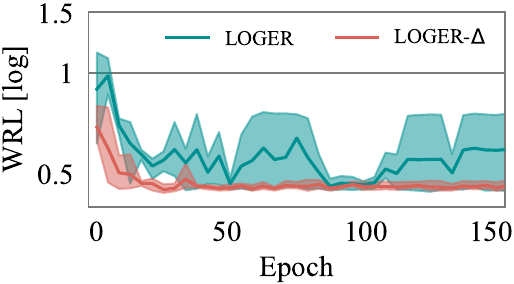}
    \end{minipage}
}
\subfigure[LOGER and LOGER$\textsf{-}\Delta$ on $\emph {Stack}$]{
    \begin{minipage}[t]{0.30\linewidth}
        \includegraphics[width=1\linewidth]{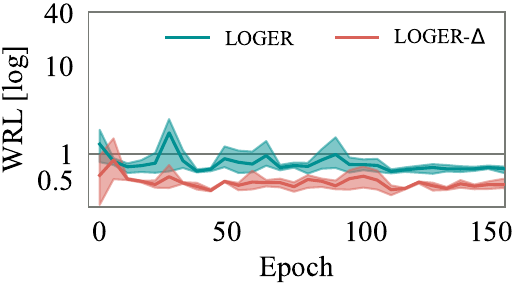}
    \end{minipage}
}\vspace*{-0.08in}

\subfigure[Balsa and Balsa$\textsf{-}\Delta$ on $\emph {JOB}$]{
    \begin{minipage}[t]{0.30\linewidth}
        \includegraphics[width=1\linewidth]{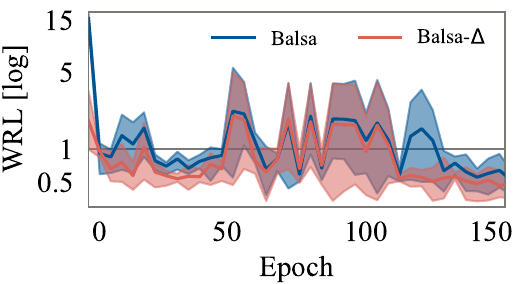}
    \end{minipage}
}
\subfigure[Balsa and Balsa$\textsf{-}\Delta$ on $\emph {TPC-DS}$]{
    \begin{minipage}[t]{0.30\linewidth}
        \includegraphics[width=1\linewidth]{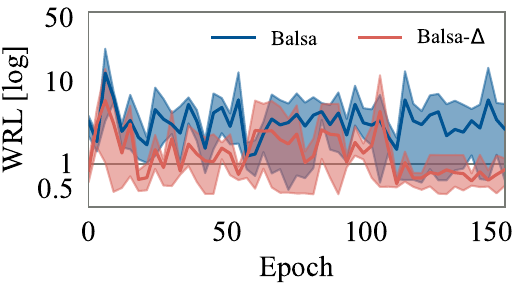}
        \vspace{-0.18in}
    \end{minipage}
}
\subfigure[Training plan size on different workloads]{
    \begin{minipage}[t]{0.30\linewidth}
        \includegraphics[width=1\linewidth]{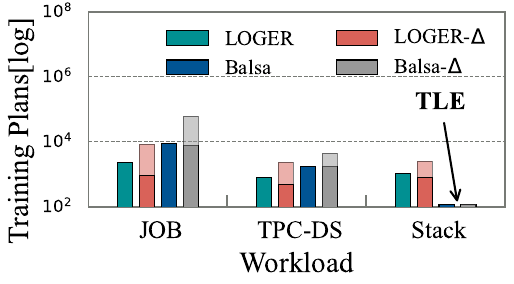}
        \label{subfig:datacost}\vspace{-0.18in}
    \end{minipage}
}
\caption{Training curves of \ours and value-based optimizers and their training plan sizes needed when convergence}
\label{fig:overall}
\vspace{-0.18in}
\end{figure*}

\begin{figure*}[th]
\centering
\subfigure[$\emph {JOB}$]{
    \begin{minipage}[t]{0.30\linewidth}
        \includegraphics[width=1\linewidth]{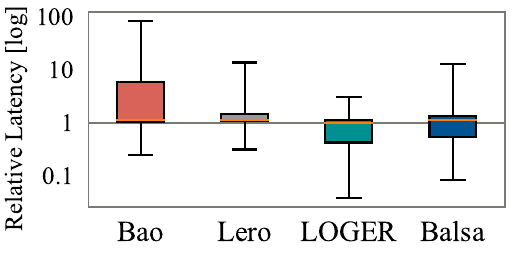}
    \end{minipage}
     \label{fig:box_job}
}
\subfigure[$\emph {TPC-DS}$]{
    \begin{minipage}[t]{0.30\linewidth}
         \includegraphics[width=1\linewidth]{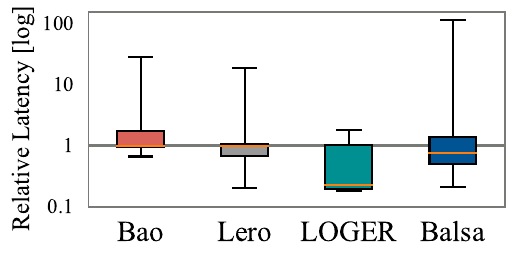}
    \end{minipage}
    \label{fig:box_tpcds}
}
\subfigure[$\emph {Stack}$]{
    \begin{minipage}[t]{0.30\linewidth}
         \includegraphics[width=1\linewidth]{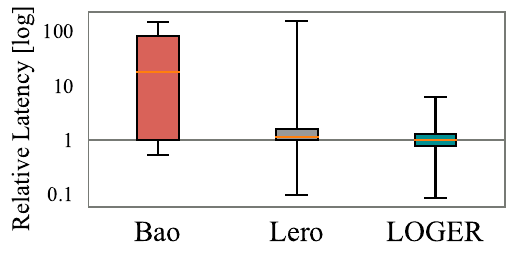}
    \end{minipage}
    \label{fig:box_stack}
}
\caption{The plan performance distribution of the candidate plan set generated by different planners}
\label{fig:boxplot}
\vspace{-0.17in}
\end{figure*}
\begin{figure*}[th]
\centering
\subfigure[WRL]{
        \includegraphics[width=0.43\linewidth]{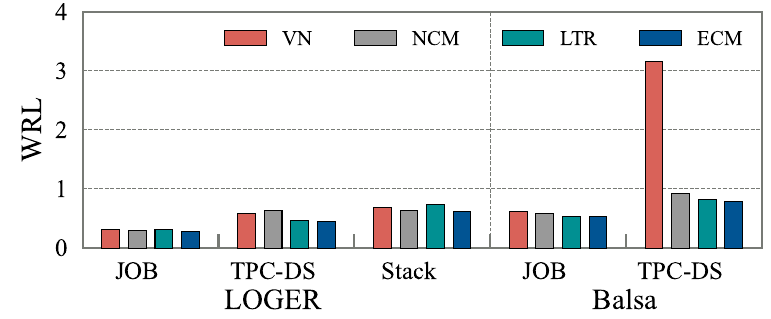}
     \label{fig:stage_two_wrl}
     \vspace{-0.19in}
}
\subfigure[GMRL]{
       \includegraphics[width=0.43\linewidth]{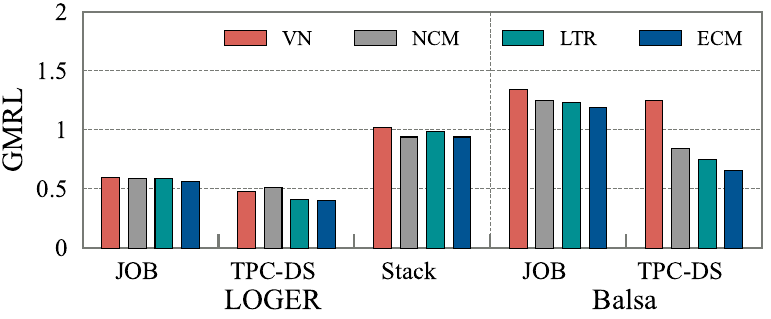}
     \label{fig:stage_two_gmrl}
     \vspace{-0.19in}
}
\caption{The performance of different models on optimal ones selection from the candidates generated by LOGER or Balsa} 
\label{fig:stage_two}
\vspace{-0.20in}
\end{figure*}
\textbf{Training cost analysis.}
Table~\ref{tab:final} shows that applying the \ours framework significantly reduces training costs for all \VBO s across the three benchmarks.
Compared to LOGER, LOGER$\textsf{-}\Delta$
converges using only 39.8\%, 34.3\%, and 69.0\% of the training time on JOB, TPC-DS, and Stack benchmarks, respectively. 
Similarly, Balsa$\textsf{-}\Delta$ requires only 93.3\% and 96.6\% of Balsa's training time on the JOB and TPC-DS benchmarks.
In comparison to learned \CBO s, although \ours may require more training time, the quality of the generated plans is substantially superior on both WRL and GMRL.  

This improvement 
is primarily
due to \textsf{Delta}’s two-stage framework.
In stage $\Rmnum{1}$, \ours transfers a part of the accuracy requirements from the value network to the cost model in stage $\Rmnum2$, which is easier to train. 
Furthermore, \ours reuses and enhances the training data in stage $\Rmnum{1}$ for training the model in stage $\Rmnum{2}$, which significantly reduces the training cost.  

\textbf{Training curve analysis.}
To better analyze the performance of \textsf{Delta},
Figure~\ref{fig:overall} shows the training curves of LOGER, LOGER$\textsf{-}\Delta$, Balsa, and Balsa$\textsf{-}\Delta$ in three workloads, along with the training data required for convergence.
Shaded areas indicate the range between minimum and maximum values over three runs and the solid lines show the median.
The curve tracks log-scale WRL over training epochs, with the gray line marking PostgreSQL’s performance.

One can observe that 
\textsf{Delta} consistently achieves faster convergence and better WRL across all workloads and two underlying \VBO s.
Notably, in early training,
LOGER and Balsa perform poorly,
while \textsf{Delta} already achieves a level of performance comparable to PostgreSQL.
LOGER$\textsf{-}\Delta$ takes only 40\%, 34\%, and 69\% of the convergence time of LOGER on \emph{JOB}, \emph{TPC-DS}, and \emph{Stack}, respectively.
Balsa$\textsf{-}\Delta$ converges with 93\% and 97\% time of Balsa on \emph{JOB}, and \emph{TPC-DS}. 
This performance improvement is particularly pronounced in the large-scale dataset \emph{Stack}.
After sufficient training, LOGER$\textsf{-}\Delta$ 
reaches a WRL 0.2 lower than LOGER.

Additionally, \textsf{Delta} exhibits more stable and reliable  training performance, showing lower variability compared to LOGER and Balsa, particularly on \emph{TPC-DS} and \emph{Stack}.
This is mainly due to the mixed cost-based planner and CQD effectively enhancing the robustness of \VBO s.

Figure~\ref{subfig:datacost} shows the number of executed plans during training until convergence.
The transparent portions of LOGER$\textsf{-}\Delta$ and Balsa$\textsf{-}\Delta$ 
indicate additional plans generated via data augmentation.
As shown, \textsf{Delta} reduces the number of executed plans required by LOGER and Balsa, greatly improving training efficiency.
Moreover, through data augmentation, \textsf{Delta} expands samples, providing sufficient training data for the model in stage $\Rmnum{2}$ and significantly minimizing cost model training expenses.

\begin{figure*}[th] \centering \subfigure[The impact of \textsf{Delta}'s elements]{     \begin{minipage}[t]{0.23\linewidth}         \includegraphics[width=1\linewidth]{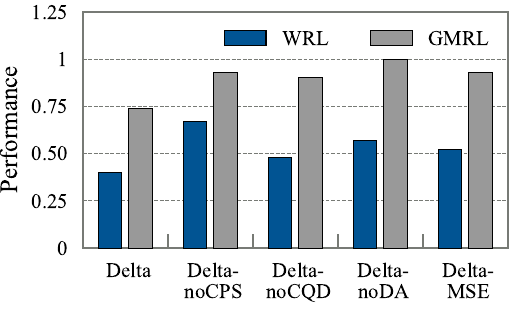}     \end{minipage}      \label{fig:abaltion} } \subfigure[The impact of $\gamma$]{     \begin{minipage}[t]{0.23\linewidth}         \includegraphics[width=1\linewidth]{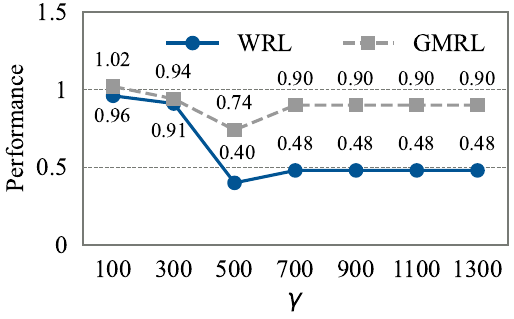}     \end{minipage}     \label{fig:gamma} } \subfigure[The impact on planning time of $k$]{     \begin{minipage}[t]{0.23\linewidth}         \includegraphics[width=1\linewidth]{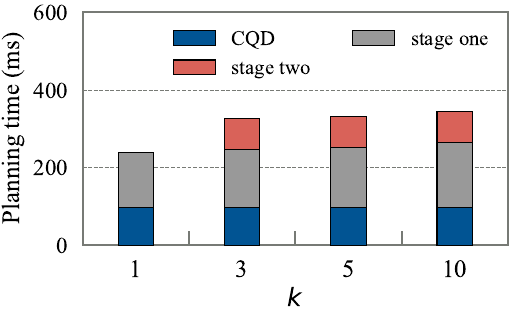}     \end{minipage}     \label{fig:kplanning} } \subfigure[The impact on performance of $k$]{     \begin{minipage}[t]{0.23\linewidth}         \includegraphics[width=1\linewidth]{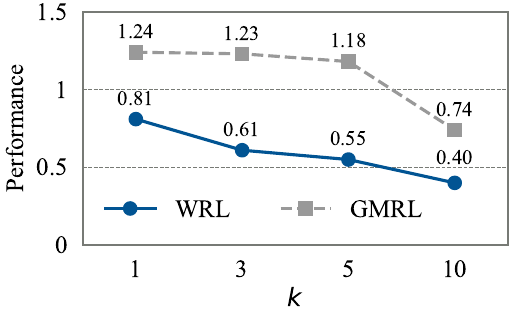}     \end{minipage}     \label{fig:kperformance} } 
\caption{The impact of \textsf{Delta}'s different choices on JOB test workload on top of Balsa} \label{fig:all} \vspace{-0.25in} \end{figure*}
\begin{figure}[th] \centering \hspace*{-0.05in} \subfigure[CQD on top of LOGER]{         \includegraphics[width=0.46\linewidth]{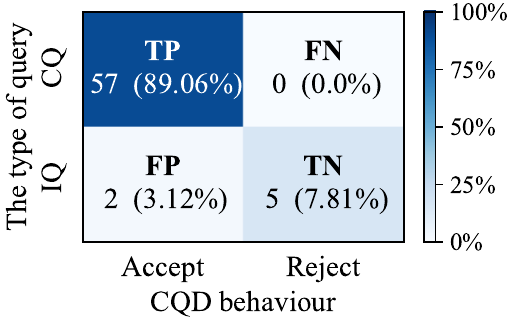}         \label{fig:cmatrix1} } \hspace*{-0.05in} \subfigure[CQD on top of Balsa]{         \includegraphics[width=0.46\linewidth]{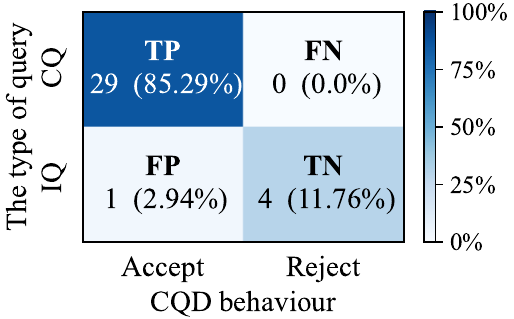}
\label{fig:cmatrix2} } \caption{Confusion matrix of CQD} \label{fig:cmatrix} \vspace{-0.25in} \end{figure}
\subsection{Value or Cost : Horses for Courses}~\label{subsec:horses}
To validate 
the value network and cost model for
plan search and optimal plan selection in \textsf{Delta}'s two-stage design, two corresponding experiments are conducted.

\textbf{Which one can generate better sets of candidate plans?}
\label{subsubsec:7.3.1}
To evaluate the effectiveness of different plan generators in exploring diverse search spaces and producing high-quality candidate plan sets, we conduct an experimental study. 
Four state-of-the-art planners are compared: the cost-based ones Bao and Lero, and the value-based ones LOGER and Balsa. 
All planners are trained on three workload training splits, though Balsa was excluded from the Stack workload due to incompatibility. 
In testing, each planner generates the candidate plan set for each query from the test workloads and their relative latency values to PostgreSQL are recorded.
Figure~\ref{fig:boxplot} summarizes the results through a boxplot of normalized relative latency values of candidate plans. 
The yellow line in each box marks the median latency of all plans produced by each planner, reflecting its central performance tendency.

The evaluation demonstrates that VBOs produce better candidate sets than CBOs,
based on three metrics.
(1) Best-case improvement: VBOs (LOGER and Balsa) consistently produce optimal plans with lower latency than CBOs (Bao and Lero) across all workloads.
(2) Worst-case risk: The worst plans generated by VBOs exhibit reduced latency degradation compared to CBOs in most scenarios.
(3) High-quality plan ratio: 
The ratio of the high-quality plans ($\mathcal{R}(P) < 1$) generated by VBOs achieves a median ratio (yellow line) (\(\eta(S_{\mathcal{G}}^{k}) \geq 50\%\)) across workloads, significantly outperforming CBOs, which frequently fall below this threshold.  
The results align with the theoretical analysis in Section~\ref{subsec:one}, confirming that 
value-guided generation in VBOs outperforms CBOs' heuristic-based approaches in producing robust candidate sets.

However, it is worth noting that in certain workloads, such as Balsa’s worst-case latency on TPC-DS, performance remains relatively poor due to the difficulty of training the value network.
Nonetheless, \textsf{Delta}’s two-stage design can mitigate this by accurate selection of the optimal plan.

\textbf{Which one can select the optimal plan more accurately?}~\label{subsubsec:7.3.2}
To evaluate
plan selection effectiveness of models,
we compare the value network (VN) with three cost models: the naive cost model (NCM), the learning-to-rank model (LTR), and our enhanced cost model (ECM),
by assessing their ability to select optimal plans from candidate sets.
All models are trained on the same executed plan pool 
until convergence. 
We then feed the candidate sets generated by VBOs from Section~\ref{subsubsec:7.3.1} into models and assess their performance on the quality of the selected plans.

Figure~\ref{fig:stage_two} shows WRL and GMRL of the final plans selected by different models on test workloads.
Overall, the performance of the three cost models is better than the value network in most scenarios.
This is because the cost model is easier to train than the value network and provides more accurate cost estimates for complete plans, making it more suitable for the plan selection.
Among them, our enhanced cost model consistently achieves the best performance across all scenarios.
This improvement results from the data augmentation and the heteroscedastic regression loss function, which together significantly boost the accuracy of the cost model.
\vspace{-0.5pt}

\subsection{The impact of CQD}\label{subsec: 7.4}
To evaluate the effectiveness of CQD, Figure~\ref{fig:cmatrix1} and Figure~\ref{fig:cmatrix2} show the confusion matrices
for
LOGER$\textsf{-}\Delta$ and Balsa$\textsf{-}\Delta$ respectively.
The matrix records the sum of experimental data from all three workloads.
Rows represent queries types, including compatible queries (CQ) and incompatible queries (IQ).
The columns show whether the CQD accepts the query to be optimized by the mixed cost-based planner. 
Each cell records the number and percentage of different categories of queries.
A desirable CQD should achieve high recall and specificity.
Higher recall means fewer compatible queries are mistakenly rejected, causing less negative impact.
Higher specificity means more incompatible queries are correctly filtered, producing more benefits.
Figure~\ref{fig:cmatrix} shows that
CQD achieves a 100\% recall rate in both cases, meaning it has no adverse effects on the planner.
Specificity reaches 71.4\% on LOGER and 80.0\% on Balsa, indicating CQD successfully prevents poor plans for incompatible queries.

\subsection{Ablation Study and Parameter Analysis} \label{subsec:ablation_and_parameter}
To evaluate the influence of architectural components and parameters on \textsf{Delta},
we conduct
systematic ablation studies and parameter sensitivity analyses.
We use the JOB workload, a widely adopted real-world benchmark in query optimization, and choose Balsa as the underlying model.
This choice is 
due to 
LOGER's full compatibility with the test queries in JOB (i.e., zero incompatible queries), 
making it unable to showcase the CQD's effectiveness.

\textbf{Ablation study.} \label{subsubsec:ablation}
We investigate the impact of different elements in \textsf{Delta} on query optimization. 
The corresponding results are shown in Figure~\ref{fig:abaltion}.
\textsf{Delta}-noCPS, \textsf{Delta}-noCQD, and \textsf{Delta}-noDA are the variants of Balsa$\textsf{-}\Delta$ without the 
CPS, CQD, and \emph{data augmentation}, respectively.
\textsf{Delta}-MSE 
replaces the heteroscedastic regression loss with mean squared error (MSE) loss during training.
We can observe that each component of \textsf{Delta} does
contribute to the query execution performance.
In particular, WRL increases 68\%, 20\%, 43\%, and 30\% without CPS, CQD, \emph{data augmentation}, and \emph{heteroscedastic regression loss function}.
GMRL increases 26\%, 22\%, 35\%, and 26\%, respectively.
CPS and data augmentation show the largest impact, confirming the rationality and effectiveness of \textsf{Delta}'s design.
\emph{Data augmentation} and \emph{heteroscedastic regression loss function} in CPS do boost performance,
while CQD also benefits query optimization.

\textbf{Parameter evaluation.} \label{subsubsec:parameter}
\emph{Effect of $\gamma$}.
To verify the effect of threshold $\gamma$ on 
CQD in \textsf{Delta}, we vary $\gamma$ from 100 to 1300.
Figure~\ref{fig:gamma} shows the results of JOB test workload on the top of Balsa. 
\textsf{Delta} gets the best performance when $\gamma$ is 500,
which is used as the default.
With smaller values (e.g., 100 or 300), 
\textsf{Delta} 
behaves conservatively,
yielding performance close to PostgreSQL.
When $\gamma$ exceeds 700, CQD accepts all queries, and performance aligns with \textsf{Delta}-noCQD.

\emph{Effect of $k$}.
Figure~\ref{fig:kplanning} and Figure~\ref{fig:kperformance} study \textsf{Delta}'s mean per-query planning time and performance on JOB (with Balsa) using different values of $k$, the number of candidate plans.
In all cases, planning time remains under 350ms.
Larger $k$ improves WRL and GMRL but slightly increases planning time, mainly due to value-based candidate plan generation.
Fortunately, the planning time is insensitive to $k$. 
 $k$ is set to 10 in our experiments and 
it can be adjusted as actual needs.


\balance 

\flushbottom
\section{Conclusion}\label{sec:conclusion}

We propose \textsf{Delta}, a mixed cost-based query optimization framework that addresses the limitations of CBOs and VBOs through a compatible query detector and a two-stage planner. 
Compared to traditional CBO, PostgreSQL, \ours achieves 2.34$\times$ faster workload execution by mitigating search space explosion and heuristic pruning constraints. 
Against two-stage CBOs, \ours outperforms by 2.11$\times$
via
higher candidate plan quality produced by the value-based planner. 
\ours also outperforms VBOs, requiring 66.6\% of training time while delivering 1.94$\times$ higher execution performance, due to its fine-grained learned cost-based ranking.
The Mahalanobis distance-based detector further ensures robustness by preemptively filtering incompatible queries. 
These results highlight \textsf{Delta}'s ability to unify the value network and the cost model and detect incompatible queries to determine more efficient plans, offering a practical solution for modern query optimization.

\clearpage
\bibliographystyle{IEEEtran}
\bibliography{main}


\end{document}